\documentclass{article}

\usepackage{fullpage}

\usepackage{amsfonts}       
\usepackage{complexity}
\usepackage[ruled,linesnumbered]{algorithm2e}
\usepackage{verbatim}
\usepackage{amsthm}
\usepackage{amsmath}
\date{}
\newlang{\Reach}{Reach}
\newcommand{\sep}{\textsf{sep}}
\newcommand{\AND}{\textbf{and}}
\newcommand{\recdec}[2]{\textsf{rdtree}(#1, #2)}
\newcommand{\indgraph}[2]{G_{\langle #1, #2\rangle}}
\newcommand{\rt}[1]{\textsf{root}(#1)}
\newcommand{\parent}[1]{\textsf{parent}(#1)}

\newcommand{\leftchild}[1]{\textsf{left}(#1)}
\newcommand{\rightchild}[1]{\textsf{right}(#1)}
\newcommand{\vanc}[1]{V^{\textsf{anc}}_{#1}}
\newcommand{\eanc}[1]{E^{\textsf{anc}}_{#1}}
\newcommand{\ganc}[1]{G^{\textsf{anc}}_{#1}}
\newcommand{\vad}[1]{V_{#1}}
\newcommand{\ead}[1]{E_{#1}}
\newcommand{\gad}[1]{G_{#1}}
\newcommand{\pos}[2]{\textsf{index}_{#1}(#2)}

\newcommand{\len}[2]{L(#1 , #2)}
\newcommand{\lseq}[2]{\textsf{SEQ}_{#1, #2}}

\newcommand{\Tdepth}[0]{h}

\newtheorem{theorem}{Theorem}
\newtheorem{lemma}{Lemma}
\newtheorem{definition}{Definition}

\title{Reachability in High Treewidth Graphs}

\author{
	Rahul Jain\\
	Department of Computer Science and Engineering\\
	Indian Institute of Techonology Kanpur\\
	\texttt{jain@cse.iitk.ac.in} \\
	\and
	Raghunath Tewari \\
	Department of Computer Science and Engineering\\
	Indian Institute of Techonology Kanpur\\
	\texttt{rtewari@cse.iitk.ac.in} \\
}

\begin{document}
	\maketitle
	
	\begin{abstract}
		Reachability is the problem of deciding whether there is a path from one vertex to the other in the graph. Standard graph traversal algorithms such as DFS and BFS take linear time to decide reachability however their space complexity is also linear. On the other hand, Savitch's algorithm takes quasipolynomial time although the space bound is $O(\log^2 n)$. Here, we study space efficient algorithms for deciding reachability that runs simultaneously in polynomial time. 
	
		In this paper, we show that given an $n$ vertex directed graph of treewidth $w$ along with its tree decomposition, there exists an algorithm running in polynomial time and $O(w\log n)$ space, that solves reachability in the graph.
	\end{abstract}

	\section{Introduction}
	Given a graph $G$ and two vertices $u$ and $v$ in $G$, the reachability problem is to decide if there exists a path from $u$ to $v$ in $G$. This problem is $\NL$-complete for directed graphs and $\L$-complete for undirected graphs  \cite{Reingold}. Hence its study gives important insight into space bounded computations. We will henceforth refer to the problem of directed graph reachability as {\Reach}. The famous open question $\L \stackrel{?}{=} \NL$ essentially asks if there is a deterministic logspace algorithm for {\Reach} or not. {\Reach} can be solved in $\Theta(n\log n)$ space and optimal time using standard graph traversal algorithms such as DFS and BFS. We also know, due to Savitch, that it can be solved in $\Theta(\log^2 n)$ space \cite{Savitch}. However, Savitch's algorithm requires $n^{\Theta(\log n)}$ time. Wigderson surveyed reachability problems in which he asked if there is an algorithm for {\Reach} that runs simultaneously in $O(n^{1 -\epsilon})$ space (for any $\epsilon >0$) and polynomial time \cite{Wigderson}. Here, we make some partial progress towards answering this question.
	
	In 1998 Barnes et al. made progress in answering Wigderson's question for general graph by presenting an algorithm for {\Reach} that runs simultaneously in $n/2^{\Theta(\sqrt{\log n})}$ space and polynomial time \cite{BBRS}. Using this result, Asano et al. improved the space bound of DFS and showed that it could be performed using $O(n)$ bits of space and $O(m\log n)$ time \cite{AsanoTaisuke}. For several other topologically restricted classes of graphs, there has been significant progress in giving polynomial time algorithms for {\Reach} that run simultaneously in sublinear space. For grid graphs a space bound of $O(n^{1/2 + \epsilon})$ was first achieved \cite{Asano11}. The same space bound was then extended to all planar graphs by Imai et al. \cite{imai}. Later for planar graphs, the space bound was improved to $\tilde{O}(n^{1/2})$ space by Asano et al. \cite{Asano}. For graphs of higher genus, Chakraborty et al. gave an $\tilde{O}(n^{2/3}g^{1/3})$ space algorithm which additionally requires, as an input, an embedding of the graph on a surface of genus $g$ \cite{ChakrabortyPavan}. They also gave an $\tilde{O}(n^{2/3})$ space algorithm for $H$ minor-free graphs which requires tree decomposition of the graph as an input and $O(n^{1/2 + \epsilon})$ space algorithm for $K_{3,3}$-free and $K_5$-free graphs. For layered planar graphs, Chakraborty and Tewari showed that for every $\epsilon >0$ there is an $O(n^{\epsilon})$ space algorithm \cite{Chakraborty}. 
	
	Treewidth is a well-studied property of graphs. The value of treewidth can range from $1$ (for a tree) to $n-1$ (for a complete graph on $n$ vertices). The computational complexity of many difficult problems become easy for bounded treewidth graphs. The weighted independent set problem can be solved in $O(2^wn)$ time \cite{Bodlaender}. Similarly, we can solve other classic problems such as the Hamiltonian circuit, vertex cover, Steiner tree, and vertex coloring in linear time for bounded treewidth \cite{Arnborg2}. It is $\NP$-complete to find on given input $\langle G, k \rangle$, if $G$ has treewidth $k$ \cite{Arnborg}. However, an $O(\sqrt{\log n})$-factor approximation algorithm is known \cite{Feige}. Series-parallel graphs are equivalent to graphs of treewidth $2$. For them, Jackoby and Tantau showed a logspace algorithm for {\Reach}. Das et al. extended the logspace bound to bounded treewidth graphs when the input contains the tree decomposition \cite{Das}. Elberfeld et al. showed a logspace algorithm for any monadic second order property of a logical structure of bounded treewidth \cite{Elberfeld}.
	\subsection*{Our Result}
	In this paper, we present a polynomial time algorithm with improved space bound for deciding reachability in graphs of treewidth $w$. In particular, we show the following result.
	
	\begin{theorem}
		\label{thm:main}
		Given a graph $G$, a tree decomposition $T$ of $G$ of treewidth $w$, and two vertices $u$ and $v$ in $G$, there is an $O(w\log n)$ space and polynomial time algorithm that decides if there is a path from $u$ to $v$ in $G$. 
	\end{theorem}
	
	To prove this, we first give an algorithm to compute a more balanced tree decomposition from the given tree decomposition that has similar treewidth. Our idea to compute this new tree decomposition is based on the idea of computing a balanced binary tree decomposition from Elberfeld et al. \cite{Elberfeld}. However, in their paper, they construct a graph called descriptor decomposition of size $\Omega(n^w)$ as an intermediate. They then use this graph to construct a tree decomposition. They bound the depth of tree decomposition by $O(w\log n)$. As a result, their algorithm requires $\Omega(n^w)$ time and the bound on depth is dependent on $w$. We, on the other hand, construct a tree decomposition using only $\poly(n)$ time and of depth $O(\log n)$. It is important to note here that the depth of our tree decomposition is independent of $w$. Next, we give an algorithm to decide reachability in $G$, which uses the newly constructed balanced binary tree decomposition. We use the idea of {\em universal sequences} from the work of Asano et al. \cite{Asano}, to bound the time complexity of our algorithm by a polynomial.  
	
	For graphs of treewidth $n^{1-\epsilon}$, for any $\epsilon > 0$, our algorithm gives an $O(n^{1-\delta})$ space and polynomial time algorithm (for some $\delta$), thus giving an affirmative answer to Wigderson's question for such class of graphs. For graphs of polylog treewidth, we show that reachability is in polynomial time and polylog space. Graphs which have genus $g$ have treewidth $O((gn)^{1/2})$, hence our algorithm gives a $O((gn)^{1/2}\log n)$ space and polynomial time algorithm for it. For planar graphs, our approach gives $O(n^{1/2} \log n)$ space and polynomial time algorithm. 
	
	Let $H$ be a graph on $h$ vertices. An $H$ minor-free graph is also, by definition, $K_h$ minor free where $K_h$ is a complete graph on $h$ vertices. Graphs which exclude a fixed minor $K_h$, have a treewidth of $h{n}^{1/2}$ \cite{nonplanar}\cite{kawarareed}. Hence, for constant $h$, our approach results in $O(n^{1/2} \log n)$ space and polynomial time algorithm for $H$ minor-free graphs.
	
	In Section \ref{sec:prelim} we give the definitions, notations and previously known results that we use in this paper. In Section \ref{sec:decomp} we show how to efficiently compute a logarithmic depth, binary tree decomposition of $G$ having a similar width from the input tree decomposition. In Section \ref{sec:reach} we give the reachability algorithm and prove its correctness and complexity bounds.
	
	\section{Preliminaries}
	\label{sec:prelim}
	For a graph $G$ on $n$ vertices, we denote its vertex and edge sets as $V(G)$ and $E(G)$ respectively.
	Let $W$ be a subset of $V(G)$. We denote the subgraph of $G$ induced by the vertices in $W$ by $G[W]$. Let $[n]$ denote the set $\{1, 2, \ldots , n\}$ for $n\geq 1$ .
	
	We next define the necessary terminology and notations related to tree decomposition that we use in this paper. For tree decomposition, we will treat the graph as an undirected graph by ignoring the direction of its edges. 
	
	For a graph $G$, a {\em tree decomposition} is a labeled tree $T$ where the labeling function $B: V(T) \rightarrow \{X \mid X \subseteq V(G) \}$ has the following property: (i) $\bigcup_{t \in V(T)} B(t) = V(G)$, (ii) for every edge $\{v, w\}$ in $E(G)$, there exists $t$ in $V(T)$ such that $v$ and $w$ are in $B(t)$, and (iii) if $t_3$ is on the path from $t_1$ to $t_2$ in $T$, then $B(t_1) \cap B(t_2) \subseteq B(t_3)$. The {\em treewidth} of a tree decomposition $T$ is $\max_{t \in V(T)} (B(t) - 1)$. Finally the treewidth of a graph $G$ is the minimum treewidth over all tree decompositions of $G$. We refer to a node $t$ in $V(T)$ as a {\em treenode} and the set $B(t)$ to be the {\em bag} corresponding to $t$.
	
	The next tool that we would be using is that of separators in graphs. For a subset $W$ of $V(G)$, a {\em vertex separator} of $W$ in $G$, denoted as $\sep(W)$, is a subset $S$ of $V(G)$ such that every component of the graph $G[V(G) \setminus S]$ has at most $\lvert W \rvert/2$ vertices of $W$.
	
	We state here the following commonly known result about vertex separators in the form that we would be using it. We also give a proof of the result for the sake of completeness.
	
	\begin{lemma} 
		\label{lem:sep}
		Let $G$ be a graph and $T$ be a tree decomposition of $G$. For every subset $U$ of $V(G)$, there exists a vertex $t$ in $V(T)$ such that the bag $B(t)$ is the vertex separator of $U$ in $G$. 
	\end{lemma}
	\begin{proof} We root the tree arbitrarily. For a treenode $t$ in $V(T)$, we denote its parent by $\parent{t}$. Let $C(t) = B(t) \cap U$. We define weights on the vertices of $T$, such that each vertex of $U$ is counted in only one of the weights. $\alpha(t) = \lvert C(t) \setminus C(\parent{t})\rvert$. Thus, $\sum_{t \in V(T)}\alpha(t) = \lvert U \rvert$. In a weighted tree, there exists a vertex whose removal divides the tree into components whose weights are at most half the total weight of the tree. Let this vertex be $t^*$ for the weight function $\alpha$. We claim that $B(t^*)$ is the vertex separator for $U$ in $G$. To prove this, we will prove that for a connected component $H$ of $G[V(G) \setminus B(t^*)]$, $H$ is a subset of $(\cup_{t \in T_i}B(t))$ for some subtree $T_i$ of $T \setminus \{t^*\}$. Since the total weight of this subtree is at most half the total weight on $T$, it would follow that the number of vertices of $U$ contained in this set would be at most half, thus proving the lemma.
		
		We will now prove that $H \subseteq (\cup_{t \in T_i}B(t))$. We first observe that a vertex $v \in V(G) \setminus B(t^*)$ can be in the bag of only one of the subtree, since otherwise, it would belong to $B(t^*)$ as well, due to the third property of tree decomposition. Now, let us assume that there are two vertices of $H$ which belong to the bags of two different subtrees, say $T_i$ and $T_j$. Since they are in connected component, there will exist a path between them. In this path, there will exist an edge, whose endpoints $v_1$ and $v_2$ would belong to different subtrees. We thus get a contradiction to the second property of tree decomposition.
	\end{proof}
	
	\section{Finding a Tree Decomposition of Small Depth}
	\label{sec:decomp}
	In this section, we show how to compute a binary tree decomposition (say $T'$) with logarithmic depth and treewidth $O(w)$. We require the tree decomposition to have logarithmic depth because our main algorithm for reachability (Algorithm \ref{algo:reach}) might potentially store reachability information for all vertices corresponding to the bags of treenodes in a path from the root to a leaf. Once the depth is reduced to $O(\log n)$ with bag size being $O(w)$, the algorithm will only need to store reachability information of $O(w\log n)$ vertices. Thus, we prove the following theorem. 
	
	\begin{theorem} 
		\label{thm:decomp}
		Given as input $\langle G, T \rangle$ where $G$ is a graph and $T$ is a tree decomposition of $G$ with treewidth $w$, there exists an algorithm working simultaneously in $O(w\log n)$ space and polynomial time which outputs a binary tree decomposition $T'$ of $G$ which has treewidth $6w + 6$ and depth $O(\log n)$.
	\end{theorem}
	
	We will now develop the framework that will help us to prove Theorem \ref{thm:decomp}.
	
	First, we show how to compute a vertex separator of a given set $U$ in $G$ in polynomial time and $O(w\log n)$ space. We cycle through every node in the tree $T$ and store the set of vertices in $B(t)$. Doing this requires $O(w\log n)$ space. Then using Reingold's undirected reachability algorithm \cite{Reingold}, we count the number of vertices of $U$ in each of the components of $G[V(G) \setminus B(t)]$. By Lemma \ref{lem:sep}, at least one of these sets $B(t)$ would be a separator of $U$ in $G$. Its size will be the size of $B(t)$ for some treenode $t$. Hence it can be at most $w+1$. We summarise this procedure in Lemma \ref{lem:compsep}.
	
	\begin{lemma}
		\label{lem:compsep}
		Given as input $\langle G, T, U \rangle$ where $G$ is a graph, $T$ is a tree decomposition of $G$ with treewidth $w$, and $U$ is a subset of $V(G)$, there exists an $O(w\log n)$ space and polynomial time algorithm that computes $\sep(U)$ and $|\sep(U)|$ is at most $w+1$.
	\end{lemma}
	
	\subsection{Constructing a Recursive Decomposition}
	As an intermediate step, we first construct a {\em recursive decomposition} of the graph which is a tree whose nodes represents a subgraph of $G$. The root node represents the entire $G$. We then remove a separator from it. We assume inductively that each of the connected components has their recursive decomposition and connect the root node to the roots of these recursive decompositions of connected components. We select a separator such that a small number of bits can encode each node. This recursive decomposition acts as an intermediate to our tree decomposition. Once we have a recursive decomposition of the graph, we add labels to each node such that it satisfies the properties of tree decomposition.
	
	\begin{definition}
		\label{def:rdtree}
		Let $Z \subseteq V(G)$ and a vertex $r \in (V(G) \setminus Z)$. Define $\indgraph{Z}{r}$ to be the subgraph of $G$ induced by the set of vertices in the connected component of $G[V(G) \setminus Z]$ which contains $r$. Define the tree $\recdec{Z}{r}$ which we call {\em recursive decomposition} as follows:
		\begin{itemize}
			\item The root of $\recdec{Z}{r}$ is $\langle Z, r \rangle$.
			\item Let $Z' = Z  \cup \sep(Z) \cup \sep(V(\indgraph{Z}{r})$ and let $r_1, \ldots, r_k$ be the lowest indexed vertices in each of the connected components of $G[(V(\indgraph{Z}{r}) \setminus Z']$. The children of the root are roots of the recursive decompositions $\recdec{Z'_i}{r_i}$ for each $i \in \{1, \ldots, k \}$, where $Z'_i$ is the set of vertices in $Z'$ that are adjacent to at least one vertex of $V(\indgraph{Z'}{r_i})$  in $G$.
		\end{itemize}
	\end{definition}
	
	We now show that for the graph $G$ the recursive decomposition tree structure has logarithmic depth, and every node in the tree can be encoded by a few bits. 
	
	\begin{lemma}
		\label{lem:rdtree_prop}
		Let $v_0$ be a vertex in $G$. Then the depth of the recursive decomposition $\recdec{\phi}{v_0}$ is at most $\log n$. Moreover, for a node $\langle Z, r \rangle$ in $\recdec{\phi}{v_0}$, we have $|Z| \leq 4w+4$.
	\end{lemma}
	\begin{proof}
		We prove a more general result that for any set of vertices $Z \subseteq V(G)$ and a vertex $r \in (V(G) \setminus Z)$, the depth of $\recdec{Z}{r}$ is at most $\log n$. By Definition \ref{def:rdtree}, we have that the set $\sep(V(\indgraph{Z}{r}))$ is a subset of $Z'$. Hence removal of $Z'$ divides the graph $\indgraph{Z}{r}$ into components each of which is of size at most half that of the size of $\indgraph{Z}{r}$. Since $r_1, \ldots, r_k$ are chosen from these components, it follows that the size of $\indgraph{Z'}{r_i}$ is at most half of $\indgraph{Z}{r}$. Additionally, in Definition \ref{def:rdtree} the sets $Z'_i$ are chosen in such a manner that the graphs $\indgraph{Z'_i}{r_i}$ and $\indgraph{Z'}{r_i}$ are equivalent. This proves that the size of the graph $\indgraph{Z}{r}$ halves at each level of the recursive decomposition. Hence $\recdec{\phi}{v_0}$ would have at most $\log n$ depth.
		
		We prove the second part of the lemma by induction on the depth of $\recdec{\phi}{v_0}$. This is trivially true for the root. Now let $\langle Z'_i, r_i \rangle$ be a child of $\langle Z, r \rangle$. Let $Z_i$ be the set of vertices of $Z \setminus \sep(Z)$ which are adjacent to at least one of the vertices of $V(\indgraph{Z'}{r_i})$ in $G$, and let $C_i$ be the unique connected component of $G[V(G) \setminus \sep(Z)]$ whose intersection with $\indgraph{Z'}{r_i}$ is not empty. Since $\sep(Z)$ is a separator of $Z$ in $G$, $C_i$ will contain at most $\lvert Z \rvert /2 $ vertices of $Z$ in it. This shows that $\lvert Z_i \rvert \leq \lvert Z \rvert /2$. By Definition \ref{def:rdtree}, we know that $\lvert Z'_i \rvert \leq \lvert Z_i \rvert + \lvert \sep(Z) \rvert + \lvert \sep(V(\indgraph{Z}{r}))\rvert$. The size of $\sep(V(\indgraph{Z}{r})) \leq w + 1$ and $\sep(Z) \leq w + 1$ by Lemma \ref{lem:compsep}. Lastly by induction $\lvert Z \rvert /2 \leq (4w + 4)/2$. Hence it follows that $|Z'_i| \leq 4w + 4$.
	\end{proof}
	
	We now show that the recursive decomposition tree corresponding to $G$ can be computed efficiently as well. To prove this, we give procedures that, given a node in the recursive decomposition tree, can compute its parent and children efficiently.
	
	\begin{algorithm}[h]
		\SetAlgoNoLine
		\DontPrintSemicolon
		\SetKwFor{For}{for}{do}{endfor}
		\SetKwFor{ForEach}{for each}{do}{endfor}
		\SetKwIF{If}{ElseIf}{Else}{if}{then}{else if}{else}{endif}
		\KwIn{$\langle G, T, v_0, Z, r\rangle$}
		\KwOut{Children of the node $\langle Z, r \rangle$ in $\recdec{\phi}{v_0}$}
		
		Compute $\sep(Z)$ using Lemma \ref{lem:compsep} \;
		Compute $\sep(V(\indgraph{Z}{r}))$ using Lemma \ref{lem:compsep} \;
		Let $Z' := Z \cup \sep(Z) \cup \sep(V(\indgraph{Z}{r}))$ \;
		\For{$v \in V(G)$}{
			\If{$v \in V(\indgraph{Z}{r})$ \AND ~$v$ is smallest indexed vertex in $\indgraph{Z'}{v}$}{
				Let $\widehat{Z} := \{v \in Z' \mid v \textrm{ is adjacent to } V(\indgraph{Z'}{v}) \textrm{ in } G\}$ \;
				Output $\langle \widehat{Z}, v \rangle$ \;
			}
		}
		\caption{Computes the children of the node $\langle Z, r \rangle$ in $\recdec{\phi}{v_0}$}
		\label{algo:child}
	\end{algorithm}
	
	Algorithm \ref{algo:child} outputs the children of $\langle Z, r \rangle$ in $\recdec{\phi}{v_0}$. Note that we don't explicitly store $\indgraph{Z}{r}$ but compute it whenever required. The separators in line $1$ and $2$ both have cardinality at most $w+1$ and can be computed in $O(w\log n)$ space and polynomial time by Lemma \ref{lem:compsep}. The cardinality of $Z$ is at most $4w +4$ by Lemma \ref{lem:rdtree_prop}. Therefore $|Z'|$ is at most $6w + 6$. The size of $\widehat Z$ computed would again be $4w+4$ by Lemma \ref{lem:rdtree_prop}. Thus the space required by Algorithm \ref{algo:child} is $O(w\log n)$.
	
	\begin{algorithm}[h]
		\SetAlgoNoLine
		\DontPrintSemicolon
		\SetKwFor{For}{for}{do}{endfor}
		\SetKwFor{ForEach}{for each}{do}{endfor}
		\SetKwIF{If}{ElseIf}{Else}{if}{then}{else if}{else}{endif}
		\KwIn{$\langle G, T, v_0, Z, r\rangle$}
		\KwOut{parent of the node $\langle Z, r \rangle$ in $\recdec{\phi}{v_0}$}
		
		Set $current := \langle \phi, v_0 \rangle$\; 
		\While{$\langle Z, r \rangle$ is not a child of $current$}{
			Let $\langle Z', r' \rangle$ be the child of $current$ such that $\indgraph{Z'}{r'}$ contains $r$\;
			Set $current := \langle Z', r' \rangle$\;
		}
		Output $current$\;
		\caption{Computes the parent of the node $\langle Z, r \rangle$ in $\recdec{\phi}{v_0}$}
		\label{algo:parent}
	\end{algorithm}
	
	Algorithm \ref{algo:parent} outputs the parent of $\langle Z, r \rangle$ in $\recdec{\phi}{v_0}$. It uses Algorithm \ref{algo:child} as a subroutine to get the children of a node in $\recdec{\phi}{v_0}$. Hence we can traverse the tree $\recdec{\phi}{v_0}$ in $O(w\log n)$ space and polynomial time. We summarize the above in Lemma \ref{lem:traverse}.
	
	\begin{lemma}
		\label{lem:traverse}
		Let $G$ be a graph, $T$ be a tree decomposition of $G$ with treewidth $w$ and $v_0$ be a vertex in $G$. Given $\langle G,T,v_0 \rangle$ and the node $\langle Z, r \rangle$ in $\recdec{\phi}{v_0}$, there exist algorithms that use $O(w\log n)$ space and polynomial time, and output the children and parent of $\langle Z, r \rangle$ respectively. As a consequence $\recdec{\phi}{v_0}$ can be traversed in $O(w\log n)$ space and polynomial time as well.
	\end{lemma}

	\subsection{Constructing a New Tree Decomposition}
	
	We now construct a new tree decomposition of $G$ from the recursive decomposition defined earlier. The new tree decomposition will have the same tree structure as that of the recursive decomposition. However, we will assign it a labeling function. The subgraph that a node of the recursive decomposition represents is a connected component obtained after removing a set of separators from $G$. The corresponding label for this node in the new tree decomposition is simply the set of separator vertices in the boundary of this subgraph together with the separator required to subdivide this subgraph further. We formalize this in Definition \ref{def:newtree}.
	
	\begin{definition}
		\label{def:newtree}
		Let $\widehat{T}$ be the tree corresponding to the recursive decomposition $\recdec{\phi}{v_0}$. For a node $\langle Z, r \rangle$ in $\recdec{\phi}{v_0}$, we define the function $\widehat{B}(\langle Z, r \rangle)$ as follows: 
		\[ \widehat{B}(\langle Z, r \rangle) :=  Z \cup ((\sep(V(\indgraph{Z}{r})) \cup \sep(Z)) \cap V(\indgraph{Z}{r})).\]
	\end{definition}
	
	We first show that $\widehat{T}$ is a tree decomposition of $G$ as well, with labeling function $\widehat{B}$.
	
	\begin{lemma}
		\label{lem:newtree}
		The tree $\widehat{T}$ defined in Definition \ref{def:newtree} along with the labeling function $\widehat{B}$, is a tree decomposition of $G$ of width $6w+6$. Moreover, the depth of $\widehat{T}$ is at most $\log n$.
	\end{lemma}
	\begin{proof}
		We claim that for a vertex $v$ in $\indgraph{Z}{r}$, there exists a vertex $\langle Z', r' \rangle$ in $\recdec{Z}{r}$ such that $\widehat{B}(\langle Z', r'\rangle)$ contains $v$. We prove this by induction on the height of the recursive decomposition $\recdec{Z}{r}$. If $\recdec{Z}{r}$ is just a single node, then $v$ is in $\sep(V(\indgraph{Z}{r}))$ by construction. Otherwise $v$ is either in $(\sep(Z) \cup \sep(V(\indgraph{Z}{r})))$ or in one of the connected components of $G[V(\indgraph{Z}{r}) \setminus (\sep(Z) \cup \sep(V(\indgraph{Z}{r})))]$. If $v$ is in $(\sep(Z) \cup \sep(V(\indgraph{Z}{r})))$, then $v$ is in $\widehat{B}(\langle Z, r \rangle)$ and we are done. Otherwise one of the children of $\langle Z, r \rangle$ will be $\langle \tilde{Z}, \tilde{r} \rangle$ such that $v$ is in $\indgraph{\tilde{Z}}{\tilde{r}}$. Now by induction hypothesis, there exists a vertex $\langle Z', r' \rangle$ in $\recdec{\tilde{Z}}{\tilde{r}}$  such that $\widehat{B}(\langle Z', r' \rangle)$ contains $v$. It follows that every vertex $v$ of $V(G)$ is contained in the label of at least one of the vertices of $\widehat{T}$, satisfying the first property of tree decomposition.
		
		We claim that for any edge $(u, v)$ in $G$ such that $\{u, v\} \subseteq V(\indgraph{Z}{r}) \cup Z$, either both $u$ and $v$ are in $\widehat{B}(\langle Z, r \rangle)$ or there exists a child $\langle Z'_i, r_i \rangle$ of $\langle Z, r\rangle$ such that $\{u, v\} \subseteq V(\indgraph{Z'_i}{r_i}) \cup Z'_i$. Since $u$ and $v$ are connected by an edge, there cannot exist any set of vertices $\widehat Z$ such that $u$ and $v$ are in different connected components of $G[V(G) \setminus \widehat Z]$. Let $Z' = Z \cup \sep(Z) \cup \sep(V(\indgraph{Z}{r}))$. If both $u$ and $v$ are in $Z'$, then they are in $\widehat{B}(\langle Z, r \rangle)$. Otherwise, let $r_i$ be the lowest indexed vertex in the connected component of $G[(V(\indgraph{Z}{r}) \setminus Z']$ which contains either of $u$ or $v$. Let $Z'_i$ is the set of vertices in $Z'$ that are adjacent to at least one of the vertices of $V(\indgraph{Z'}{r_i})$  in $G$. Now, if both $u$ and $v$ are not in $V(\indgraph{Z'_i}{r_i})$, then one of them have to be in $Z'_i$. Hence in all cases, $u$ and $v$ are contained in $V(\indgraph{Z'_i}{r_i}) \cup Z'_i$. Hence by induction on the height of the tree decomposition $\widehat{T}$ we have that there exists a treenode in $\widehat{T}$ whose bag contains both $u$ and $v$, satisfying the second property of tree decomposition.
		
		To establish the third property of tree decomposition we first show that if $v$ is not in $Z \cup V(\indgraph{Z}{r})$, then for no descendant $\langle \tilde{Z}, \tilde{r} \rangle$ of $\langle Z, r \rangle$ will $\widehat{B}(\langle \tilde{Z}, \tilde{r} \rangle)$ contain $v$. We show this by induction on the height of the recursive decomposition. If there is only one node in $\recdec{Z}{r}$, then $\widehat{B}(\langle Z, r \rangle)$ does not contain $v$ by definition. Otherwise, no connected component of $G[V(\indgraph{Z}{r}) \setminus Z']$ contains $v$. Also $Z'_i$ for any of its children will not contain $v$ as claimed. 
		
		Now let $\langle Z, r \rangle$ be a treenode in $\widehat{T}$. We claim that for any child $\langle Z'_i, r_i \rangle$ of $\langle Z, r \rangle$ if a vertex $v$ is in $\widehat{B}(\langle Z, r \rangle)$, then either $v$ is also in $\widehat{B}(\langle Z'_i, r_i \rangle)$ or no descendant of $\langle Z'_i, r_i \rangle$ has a bag corresponding to it which contains $v$. Let $\langle Z, r \rangle$ be a node in $\recdec{\phi}{v_0}$ such that $\widehat{B}(\langle Z, r \rangle) = X$. Since any connected component of $G[V(\indgraph{Z}{r}) \setminus \widehat{B}(\langle Z, r \rangle)]$ cannot contain $v$, $v$ is not in $V(\indgraph{Z'_i}{r_i})$ for any child $\langle Z'_i, r_i \rangle$ of $\langle Z, r \rangle$. Now if $v$ is not in $\widehat{B}(\langle Z'_i, r_i \rangle)$, then it implies that $v$ is not in $Z'_i \cup V(\indgraph{Z'_i}{r_i})$ as well. Hence the third property of tree decomposition is satisfied as well.
		
		For a vertex $\langle Z, r \rangle$ in $\recdec{\phi}{v_0}$, we have $\lvert Z \rvert \leq 4w + 4$, $ \sep(Z) \leq w + 1$ and $\sep((V(\indgraph{Z}{r}))  \leq w + 1$ as well. Hence $\widehat{B}(\langle Z, r \rangle) \leq 6w + 6$.
		
		Since the tree $\widehat{T}$ and $\recdec{\phi}{v_0}$ have the same structure, the bounds on their depths are the same.
	\end{proof}
	
	Next, we observe that given $\langle Z, r \rangle$, we can compute $\widehat{B}(\langle Z, r \rangle)$ in $O(w\log n)$ space and polynomial time. Hence we have the following Lemma.
	
	\begin{lemma}
		\label{lem:newtreetraverse}
		Given a graph $G$ and a tree decomposition $T$ of $G$ with treewidth $w$,  there is an algorithm that can compute a new tree decomposition $\widehat{T}$ of $G$ having treewidth at most $6w+6$ and depth at most $\log n$, using $O(w\log n)$ space and polynomial time. Moreover, the tree $\widehat{T}$ can be traversed in $O(w\log n)$ space and polynomial time as well.
	\end{lemma} 
	
	Note that the tree $\widehat{T}$ might not be a binary tree since a separator might disconnect the graph into more than two components. However, to decide reachability in the later part of this paper, we require the tree decomposition to have bounded degree as well. We achieve this by using the following lemma from Elberfeld, Jakoby, and Tantau to get the required tree decomposition $T'$.
	
	\begin{lemma}\cite{Elberfeld}
		\label{lem:elberfeld}
		There is a logspace algorithm that on the input of any logarithmic depth tree decomposition of a graph $G$ outputs a logarithmic depth, binary balanced tree decomposition of $G$ having the same treewidth.
	\end{lemma}
	
	Now combining Lemma \ref{lem:newtreetraverse} and Lemma \ref{lem:elberfeld} we get the proof of Theorem \ref{thm:decomp}.
	
	We observe here that the input tree decomposition $T$ is used only to compute a vertex separator in $G$. For graphs where vertex separators can be computed efficiently, such as planar graphs \cite{imai} or constant treewidth graphs, the requirement of input tree decomposition can be waived.
	
	\section{Deciding Reachability using a Balanced Binary Tree Decomposition}
	\label{sec:reach}
	In this section, we show that given a graph $G$ along with a binary, balanced tree decomposition $T$ whose depth is $O(\log n)$; there exists an efficient algorithm to decide reachability in $G$ in $O(w\log n)$ space and polynomial time. In particular, we show the following theorem.
	
	\begin{theorem} 
		\label{thm:reach}
		Given $\langle G, T, u, v \rangle$ as input, where $G$ is a graph on $n$ vertices and $T$ is a binary balanced tree decomposition of $G$ having depth $\Tdepth$ and treewidth $w$, there exists an $O(w\Tdepth + \log n)$ space and polynomial time algorithm that solves reachability in $G$.
	\end{theorem}
	
	We first state the notation required to prove Theorem \ref{thm:reach}. This notation is commonly used to describe dynamic programming algorithms which use tree decomposition. Let $T$ be a rooted binary tree. We denote $\rt{T}$ to be the root of $T$ and for a node $t \in T$, we denote $\leftchild{t}$ and $\rightchild{t}$ to be the left and right child of $t$ respectively (the value is NULL if a child does not exist). For two nodes $t$ and $t'$ in $T$, if $t'$ lies in the path from $\rt{T}$ to $t$, then we say that $t'$ is an {\em ancestor} of $t$ and $t$ is a {\em descendent} of $t'$. For a treenode $t$, let $B_e(t)$ denote the set of edges of $G$ whose both endpoints are in $B(t)$. We define a subgraph of $G$ with respect to the treenode $t$ consisting of the ancestor vertices of $t$. Formally, the vertex set is $\vanc{t} = \bigcup_{\{ t' \textrm{ is an ancestor of }t\}}B(t')$, the edge set is $\eanc{t} = \bigcup_{\{ t' \textrm{ is an ancestor of }t\}}B_e(t')$ and the graph $\ganc{t} = (\vanc{t}, \eanc{t})$. Now, we define a subgraph of $G$ with respect to the treenode $t$ consisting of the ancestor as well as descendent vertices of $t$. Formally, the vertex set is $\vad{t} = \bigcup_{\{ t' \textrm{ is an ancestor or descendent of }t\}}B(t')$, the edge set is $\ead{t} = \bigcup_{\{ t' \textrm{ is an ancestor or descendent of }t\}}B_e(t')$ and the graph $\gad{t} = (\vad{t}, \ead{t})$.
	
	We assume that the vertices $u$ and $v$ are in $\rt{T}$, for otherwise, we can add them in all of the bags of the given tree decomposition. Also, we assume that $n$ is a power of $2$.
	
	We now explain our algorithm. Let $S_1, S_2, \ldots S_m$ be subsets of $V(G)$ which covers all the vertices. Given a sequence of these, say $\langle S_{i_1}, S_{i_2}, \ldots, S_{i_k} \rangle$ we can check if there is a path from $u$ to $v$ in $G$ corresponding to this sequence in the following way. Begin with a bit-vector where each bit corresponds to a vertex in $S_1$. For all vertices $w$ in $S_1$ we check if there exists an edge $(u, w)$ in $G$ and set the bit corresponding to $w$ to $1$ in the bit-vector. We set the bit corresponding to $u$ to $1$ too if it is present in $S_1$. The marked vertices of $S_1$ have a path from $u$ to them. Now, using this vector, we mark vertices in $S_2$ which are reachable from the marked vertices in $S_1$ using at most $1$ edge in the second bit-vector. In general, given a bit-vector with marked vertices corresponding to $j$'th element in the sequence, we use the other bit-vector to mark vertices in $S_{j+1}$. We iterate $k$ times to mark the vertices in $S_{i_k}$.
	
	We say that a sequence $\langle S_{i_1}, S_{i_2}, \ldots, S_{i_k} \rangle$ is $good$ for a possible path $p$ from $u$ to $v$ when the existence of path $p$ in the graph implies that bit corresponding to $v$ is set to $1$ after the execution of the above discussed procedure. Our algorithm uses the tree decomposition to divide the vertices into sets and then finds a sequence that is $good$ for all the possible paths. The sets of vertices that our algorithm uses are $V_t$ for all the leaves $t$ of $T$. We can map a sequence of sets of vertices to a sequence of leaves of the graph. We thus need a sequence of leaves that is $good$ for all possible paths from $u$ to $v$. In Section $\ref{subsec:goodseq}$, we show how to construct a $good$ sequence of leaves. Then, in Section $\ref{subsec:solreach}$, we show how to use this sequence to solve reachability.
	
	\subsection{Constructing a Good Sequence of Leaves}
	\label{subsec:goodseq}
	We will be using $universal$ $sequences$ and the following lemma about it from Asano et al. to construct the sequence of leaves.
	
	For every integer $s \geq 0$, a {\em universal sequence} $\sigma_s$ of length $2^{s+1}-1$ is defined as follows: 
	\[ \sigma_s = \begin{cases}     \langle 1 \rangle & s = 0 \\
	\sigma_{s-1} \diamond \langle 2^s \rangle \diamond \sigma_{s-1} & s > 0
	\end{cases}
	\]
	where $\diamond$ is the concatenation operation.
	
	\begin{lemma}\cite{Asano}
		\label{lem:us}
		The universal sequence $\sigma_s$ satisfies the following properties:
		\begin{itemize}
			\item[-] Let $\sigma_s = \langle c_{1}, \ldots, c_{2^{s+1} -1} \rangle$. Then for any positive integer sequence $\langle d_1, \ldots, d_x \rangle$ such that $\Sigma d_i \leq 2^s$, there exists a subsequence $\langle c_{i_1}, \ldots, c_{i_x} \rangle$ such that $d_j \leq c_{i_j}$ for all $j \in [x]$.
			\item[-] The sequence $\sigma_s$ contains exactly $2^{s-i}$ appearances of the integer $2^i$ and nothing else.
			\item[-] The sequence $\sigma_s$ is computable in $O(2^s)$ time and $O(s)$ space.
		\end{itemize}
	\end{lemma}

	\begin{definition} Let $T$ be a balanced binary tree. Let $t$ be a node in $T$ and $d$ be a positive power of $2$.
		We define a sequence of leaves of $T$ in the following way.
		
		\[ \lseq{t}{d} = \lseq{\leftchild{t}}{c_1} \diamond \lseq{\rightchild{t}}{c_1} \diamond \lseq{\leftchild{t}}{c_2} \diamond \lseq{\rightchild{t}}{c_2} \diamond \cdots \diamond \lseq{\rightchild{t}}{c_{2d-1}} \]
		
		where $c_i$ is the $i$-th integer in $\sigma_{\log d}$. For the base case, if $t$ is a leaf, we have $\lseq{t}{d}$ is $\langle t \rangle$ concatenated with itself $d$ times. We also define $\lseq{t}{d}(r)$ to be the $r$th leaf in the sequence $\lseq{t}{d}$. The length of $\lseq{t}{d}$ is the number of leaves in $\lseq{t}{d}$.
	\end{definition}

	We now wish to show that we can construct the sequence $\lseq{\rt{T}}{d}$ in $O(h + \log d)$ space. We will first bound its length and then present a method to find its $r$th leaf.
	
	\begin{lemma}Let $T$ be a balanced binary tree. Let $t$ be a node in $T$, $d$ be a positive power of $2$ and $h$ be the height of subtree of $T$ rooted at $t$. Then, the length of sequence $\lseq{t}{d}$ is 
		
		\[2^{\Tdepth}d\binom{\Tdepth + \log d}{\log d}\] 
		
		\label{lem:bsize}
	\end{lemma}
	
	\begin{proof}
		Let $\len{\Tdepth}{d}$ be the length of the sequence $\lseq{t}{d}$. By definition of $\lseq{t}{d}$, we have 
		\[\len{\Tdepth}{d} = \begin{cases} 2\sum_{c \in \sigma_{\log d}}\len{\Tdepth - 1}{c} & \Tdepth > 0 \\
		d & h = 0\end{cases}\]
		From lemma $\ref{lem:us}$, we get that $\sigma_{\log d}$ contains exactly $\frac{d}{2^i}$ occurrences of the integer $2^i$. Thus we have:
		\[\len{\Tdepth}{d} = \begin{cases} \sum_{i = 0}^{\log d}\frac{d}{2^{i-1}}\len{\Tdepth - 1}{c} & \Tdepth > 0\\
		d & \Tdepth = 0 \end{cases}\]
		
		We claim that $\len{\Tdepth}{d} = 2^\Tdepth d\binom{\Tdepth + \log d}{\log d}$ and we prove this through induction on $\Tdepth$. For $\Tdepth = 0$, we see that
		
		\begin{align*}
		2^{\Tdepth}d\binom{\Tdepth + \log d}{\log d} &= d\binom{\log d}{\log d}\\
		&= d\end{align*}
		Now, we assume the statement to be true for smaller values of ${\Tdepth}$.
		We see that:
		
		\begin{align*}
		\len{\Tdepth}{d} &= \sum_{i = 0}^{\log d}\frac{d}{2^{i-1}}\len{\Tdepth - 1}{2^i} \\
		\len{\Tdepth}{d} &= \sum_{i = 0}^{\log d}\frac{d}{2^{i-1}}2^{\Tdepth - 1}2^i\binom{\Tdepth + i - 1}{i}\\
		\len{\Tdepth}{d} &= 2^{\Tdepth}d\sum_{i = 0}^{\log d}\binom{\Tdepth + i - 1}{i}
		\end{align*}
		using $\binom{a}{r} = \binom{a+1}{r} - \binom{a}{r-1}$
		
		\begin{align*}
		\len{\Tdepth}{d} &= 2^{\Tdepth}d\sum_{i = 0}^{\log d}(\binom{\Tdepth + i}{i} - \binom{\Tdepth + i - 1}{i-1}) \\
		\len{\Tdepth}{d} &= 2^{\Tdepth}d\binom{\Tdepth + \log d}{\log d}
		\end{align*}

	\end{proof}
	
	\begin{algorithm}[h]
		
		\SetAlgoNoLine
		\DontPrintSemicolon
		\SetKwFor{For}{for}{do}{endfor}
		\SetKwFor{ForEach}{for each}{do}{endfor}
		\SetKwIF{If}{ElseIf}{Else}{if}{then}{else if}{else}{endif}
		\KwIn{$\langle t, d, r\rangle$}
		\While{$t$ is not a leaf}{
			Let $i^*$ be the smallest integer such that $(r - 2\sum_{i = 1}^{i^*}\len{m/2}{c_i}) \leq 0$ where $c_i$ is the $i$'th integer in the sequence $\sigma_{\log d}$\;
			\eIf{$r - 2\sum_{i = 1}^{i^* - 1}\len{m/2}{c_i} - \len{m/2}{c_{i^*}} \leq 0$}{
				$r \gets r - 2\sum_{i = 1}^{i^* - 1}\len{m/2}{c_i}$\;
				$t \gets \leftchild{t}$\;
				$d \gets c_{i^*}$
			}{
				$r = r - 2\sum_{i = 1}^{i^* - 1}\len{m/2}{c_i} - \len{m/2}{c_{i^*}}$\;
				$t \gets \rightchild{t}$\;
				$d \gets c_{i^*}$
			}
			
		}
		return $t$
		\caption{Computes the $r$-th element of the sequence $\lseq{t}{d}$}
		\label{algo:bstring}
	\end{algorithm}
	
	\begin{lemma}
		Let $T$ be a binary balanced tree of depth at most $\Tdepth$. Let $t$ be a node of $T$ and $d$ be a power of $2$. The sequence $\lseq{t}{d}$ can be constructed in space $O(\Tdepth + \log d)$.
	\end{lemma}
	
	\begin{proof}
		We see that $L(m, d)$ is atmost a polynomial in $m$ and $d$. For a given integer $r$, let $i^*$ be the smallest integer such that $r - 2\sum_{i = 1}^{i^*}\len{m/2}{c_i} \leq 0$. By the definition, $\lseq{t}{d}(r) = \lseq{\leftchild{t}}{c_{i^*}}(r - 2\sum_{i = 1}^{i^* - 1}\len{m/2}{c_i})$  if $r - 2\sum_{i = 1}^{i^* - 1}\len{m/2}{c_i} - \len{m/2}{c_{i^*}} \leq 0$ and $\lseq{t}{d}(r) = \lseq{\rightchild{t}}{c_{i^*}}(r - 2\sum_{i = 1}^{i^* - 1}\len{m/2}{c_i} - \len{m/2}{c_{i^*}})$ otherwise. 
		
		The length of the sequence $\lseq{t}{d}$ is at most $2^hd\binom{h + \log d}{\log d}$. Hence the number of bits required to store any index of the sequence is at most $\log (2^hd\binom{h + \log d}{\log d}) = O(h + \log d)$. This gives the space bound of Algorithm \ref{algo:bstring}.
	\end{proof}
	
	\subsection{Algorithm to Solve Reachability}
	\label{subsec:solreach}
	
	\begin{algorithm}[h]
		\caption{Reach($G$, $T$, $u$, $v$)}
		\label{algo:reach}
		\SetAlgoNoLine
		\DontPrintSemicolon
		\SetKwFor{For}{for}{do}{endfor}
		\SetKwFor{ForEach}{for each}{do}{endfor}
		\SetKwIF{If}{ElseIf}{Else}{if}{then}{else if}{else}{endif}
		\KwIn{$\langle G, T, v , u\rangle$}
		Let $R_0$ be and $R_1$ be two $w\Tdepth$ bit-vectors\;
		Let $t_0$ and $t_1$ be two leaves of $T$ initialized arbitrarily\;
		Initialize all the bits of $R_0$ with $0$ and mark $u$ (by setting the bit at position $\pos{t_0}{u}$ to 1)\;
		\For{every leaf $f$ in $\lseq{\rt{T}}{n}$ in order}{ \label{alg:for}
			Let the iteration number be $i$\;
			Reset all the bits of $R_{i \bmod 2}$ to $0$\;
			Let $t_{i \bmod 2} \gets f$\;
			\For{all $x$ marked in $R_{(i-1) \bmod 2}$ and all $y$ in $V_f$}{
				\If{(x, y) is an edge in $G$ OR $x = y$}{
					Mark $y$ in $R_{i \bmod 2}$ (by setting the bit at position $\pos{t_{i \bmod 2}}{y}$ to 1)\;
				}
			}
		} \label{alg:endfor}
		If $v$ is marked return $1$; otherwise return $0$.
	\end{algorithm}
	
	For a leaf $t$ of $T$ and a vertex $v$ of $G$ we use $\pos{t}{v}$ for the position of $v$ in an arbitrarily fixed ordering of the vertices of $G_t$.
	
	\begin{lemma} Let $G$ be a graph and $T$ be a binary tree decomposition of $G$ of width $w$ and depth $\Tdepth$. Let $t$ be a node of $T$ and $d$ be a power of $2$. For each vertex  $y \in \vanc{t}$, $y$ is marked after the execution of iterations in lines $ \ref{alg:for}$ to $\ref{alg:endfor}$ of Algorithm $\ref{algo:reach}$ with values of $f$ in $\lseq{t}{d}$ if and only if there is a marked vertex $x$ in $\vanc{t}$ and a path from $x$ to $y$ in $\gad{t}$ of length at most $d$.
		\label{lem:iter}
	\end{lemma}
	
	\begin{proof} We prove this by induction on the height of subtree rooted at $t$. The base case is trivial. Let $p$ be the path of length at most $d$ from $x$ to $y$ such that $x$ is marked and $x, y$ is in $\vanc{t}$. We see that the edges of path $p$ will belong to either $\ead{\leftchild{t}}$ or $\ead{\rightchild{t}}$ (or both). We label an edge of $p$ as $0$ if it belongs to $\ead{\leftchild{t}}$, else label it as $1$. Break down $p$ into subpaths $p_1, \ldots, p_k$ such that the edges in $p_i$ all have same label and label of edges in $p_{i+1}$ is different form label of $p_i$. The endpoints $y_i$ of these subpaths will belong to $\vanc{t}$, for otherwise $y_i$ will not be in $B(t)$ but since $y_i$ has edges of both labels incident on it, it will be in bags of both subtrees rooted at $\leftchild{t}$ and $\rightchild{t}$ contradicting the third property of tree decomposition. Let $l_i$ be the length of path $p_i$. Since $l_1 + l_2 + \cdots + l_k \leq d$, by Lemma $\ref{lem:us}$, there exists a subsequence $\langle c_{i_1}, c_{i_2}, \ldots, c_{i_k} \rangle$ of $\sigma_{\log d}$ such that $l_j \leq c_{i_j}$.
		
		Consider the subsequence $\lseq{\leftchild{t}}{c_{i_1}} \diamond \lseq{\rightchild{t}}{c_{i_1}} \diamond \lseq{\leftchild{t}}{c_{i_2}} \diamond \lseq{\rightchild{t}}{c_{i_2}} \diamond \lseq{\leftchild{t}}{c_{i_3}} \diamond \lseq{\rightchild{t}}{c_{i_3}} \diamond \cdots \diamond \lseq{\leftchild{t}}{c_{i_k}} \diamond \lseq{\rightchild{t}}{c_{i_k}}$ of $\lseq{t}{d}$. We claim that $y_j$ is marked after the iterations with the value of $f$ in $\lseq{\leftchild{t}}{c_{i_j}} \diamond \lseq{\leftchild{t}}{c_{i_j}}$. Since $y_{j-1}$ is marked before the iterations and the path $p_{j}$ is either the subgraph $G_{\leftchild{t}}$ or $G_{\rightchild{t}}$ having length at most $c_{i_j}$, $y_j$ will be marked by induction hypothesis. 
	\end{proof}
	
	\begin{lemma}
		\label{lem:algo}
		On input of a graph $G$ with $n$ vertices and its tree decomposition $T$ with treewidth $w$ and depth $\Tdepth$; Algorithm {\ref{algo:reach}} solves reachability in $G$ and requires $O(w\Tdepth + \log n)$ space and time polynomial in $2^{\Tdepth}, n$ and $w$.
	\end{lemma}
	\begin{proof} The proof of correctness of the algorithm follows from Lemma $\ref{lem:iter}$ and the fact that $u$ and $v$ are both present in $B(\rt{T})$ and $u$ is marked before the first iteration of the for-loop in line $\ref{alg:for}$.
		
		We first analyse the space required. The size of bit-vectors $R_0$ and $R_1$ is $w\Tdepth$. $t_0$ and $t_1$ are indices of nodes of $T$. The space required to store index of a vertex of $T$ is $O(\Tdepth)$. Space required to store a vertex of $G$ is $O(\log n)$, and $\pos{t}{x}$ for a node $t$ and a vertex $x$ can be found in $O(\log n + \Tdepth)$ space. Hence the total space required is $O(w\Tdepth + \log n)$.
		
		We now analyze the time bound. By Lemma $\ref{lem:bsize}$, the size of $\lseq{t}{d}$ is polynomial in $2^h$ and $d$, the number of iterations in the for-loop of line $\ref{alg:for}$ is thus a polynomial. The other lines do trivial stuff and hence the total running time of the algorithm is polynomial.
	\end{proof}
	
	Now, Theorem $\ref{thm:reach}$ follows from Lemma $\ref{lem:algo}$. Combining Theorem $\ref{thm:reach}$ and Theorem $\ref{thm:decomp}$ we get the proof of Theorem $\ref{thm:main}$.

	\bibliographystyle{plain}
	\bibliography{references}
	
\end{document}